\documentclass{article}

\usepackage{amsmath,amsfonts,amssymb,times,graphicx,natbib,algorithm,algorithmic,hyperref}
\usepackage{microtype}
\usepackage{graphicx}
\usepackage{subfigure}
\usepackage{booktabs} % for professional tables
\usepackage{amsmath}
\usepackage{amsfonts}
\usepackage{algorithm}
\usepackage{algorithmic}
\usepackage{amssymb}
\usepackage{amsthm}
\usepackage{bbm}
\usepackage[table]{xcolor}
\usepackage{titling}
\usepackage{comment}
\usepackage[inline]{enumitem}
\usepackage{caption}
\usepackage{natbib}
\usepackage{longtable}

\def\R{{\mathbb{R}}}
\def\pr{{\rm Pr}}

\def\Z{{\mathbb Z}}
\def\E{{\mathbb E}}

\def\G{{\mathcal G}}

\newtheorem{thm}{Theorem}
\newtheorem{lemma}[thm]{Lemma}

\definecolor{light-gray}{gray}{0.9}
\allowdisplaybreaks

\usepackage[accepted]{hill2019}

\icmltitlerunning{Human Interaction and Interpretability Paper}

\begin{document}

\twocolumn[
\icmltitle{Interactive Topic Modeling with Anchor Words}


\icmlsetsymbol{equal}{*}

\begin{icmlauthorlist}
\icmlauthor{Sanjoy Dasgupta}{ucsd}
\icmlauthor{Stefanos Poulis}{ucsd,ntent}
\icmlauthor{Christopher Tosh}{columbia}
\end{icmlauthorlist}

\icmlaffiliation{ucsd}{University of California, San Diego}
\icmlaffiliation{columbia}{Columbia University}
\icmlaffiliation{ntent}{NTENT}
\icmlcorrespondingauthor{Stefanos Poulis}{spoulis@eng.ucsd.edu}

\icmlkeywords{interpretability, transparency}

\vskip 0.3in
]


\printAffiliationsAndNotice{}

\begin{abstract}
The formalism of {\it anchor words} has enabled the development of fast topic modeling algorithms with provable guarantees. In this paper, we introduce a protocol that allows users to interact with anchor words to build customized and interpretable topic models. Experimental evidence validating the usefulness of our approach is also presented. 
\end{abstract}

\section{Introduction}

Topic models~\citep{blei2002latent,GS04,li2006pachinko,blei2007correlated} are an unsupervised approach to modeling textual data. Given a corpus of documents, topic modeling seeks a small number of probability distributions over the vocabulary, called {topics}, so that each document is well-summarized as a mixture of topics. There are now several algorithms that assume data is generated by a collection of topics and aim to provably recover these topics~\citep{arora2012learning, arora2013practical}.

On the other hand, the natural interpretation of topics, that they represent the main themes of a corpus, has perhaps most motivated their use by practitioners~\citep{mimno2011optimizing,lee2014low,lund2017tandem}. Indeed, the most common way to summarize topics is with a short list of their most probable words, and topic models are judged according to how well these lists align with a user's intuition~\citep{chang2009reading}. In this sense, a user expects to interpret a topic model via a small collection of words.

%These two objectives in topic modeling, data summarization and model interpretability, often stand in opposition to each other. To see this, note that topics are probability distributions over thousands or even tens of thousands of words. It is not surprising then that even in the best of scenarios, a topic only puts a small percentage of its probability mass on the top 20 or so words. Thus, two topics may look superficially similar based on their top words, but in actuality they represent vastly different distributions over words. Moreover, if a user is only going to be presented with the top words of these topics anyway, is there anything to be gained by keeping around both of these topics?

Model fit and interpretability form two, sometimes opposing, objectives in topic modeling, and it can be difficult to strike a balance between the two. Consider the challenge of granularity: should there be different topics for `football,' `Olympics,' and `basketball' or is a single topic over `sports' sufficient? Obviously more topics will be able to describe the corpus more easily, but a particular user may not care to make the distinction between three sports-related topics. Clearly no unsupervised method can be expected to always make the correct choice here.

To deal with such inherent ambiguities, researchers have considered methods to inject user knowledge directly into topic modeling algorithms. The most common approach has been encoding positive and negative word correlations into prior distributions~\citep{andrzejewski2009incorporating, hu2014interactive, petterson2010word}. The idea is that by biasing models to group words a user knows should be together and separate words a user knows should remain apart, an algorithm can converge on a topic model that better reflects a user's preferences.

Although word correlation-based approaches lend themselves to clean probabilistic models, they require a large amount of feedback to converge on the topic model a user has in mind. Indeed, interaction that operates on the words of a vocabulary alone generally does not have the {leverage} to trigger the large changes in a topic model that a user may be hoping for. Moreover, if a user has in mind a specific deficiency in the topic model, for instance, if they want to get rid of a `low quality' topic, it may not be clear how to construct the proper word correlations to correct the issue.

A more recent approach to interactive topic modeling, which we again take up here, was proposed by \citet{lund2017tandem}. In their work, users inject knowledge and preferences into a model through anchor words -- words which only occur with significant probability in a single topic~\citep{arora2012learning,arora2013practical}. Because anchor words occur only in a single topic, users can treat them as proxies for entire topics, allowing large changes in a topic model with only a few interactions.

In this work, we present an intuitive interactive protocol wherein users are shown anchor words and are given the opportunity to group some of them, if they should belong to the same topic, while removing others that are uninteresting. A topic is then created for each group. Despite its simplicity, this interactive approach is surprisingly effective, enjoying advantages that purely unsupervised approaches lack. %We have designed the interaction to be efficient in its use of human feedback by reducing the number of anchor words a user must examine to create a group.

%\chris{What is the formalization?}

%The rest of the paper is organized as follows. In Section~\ref{sec:related-work} we review previous work on interactive approaches to traditionally unsupervised tasks. In Section~\ref{sec:preliminaries} we present some background information on anchor word concepts and methods. In Section~\ref{sec:interactive-protocol} we illustrate the potential pitfalls of a previous proposal for interacting with anchor words before presenting our own interactive anchor word-based protocol, along with a generative model it is capable of recovering. Finally, in Section~\ref{sec:experiments}, we present a series of experiments demonstrating the efficacy of our approach. 

\section{Related work}\label{sec:related-work}

The observation that unsupervised learning objectives rarely align completely with a user's intentions is not a new one. Nor is the solution of introducing human feedback to mitigate this issue. The approaches that have been studied thus far can be generally broken into two categories: constraint-based and higher-order.

In constraint-based interactive learning, a structure is found by optimizing some cost function subject to certain constraints. %In flat clustering, for example, these constraints are pairs of data points which either must belong to the same cluster ({\tt must-link}) or cannot belong to the same cluster ({\tt cannot-link})~\citep{wagstaff2001constrained,ashtiani2016clustering}. For hierarchical clustering, these constraints take the form of triplets of data points $(\{x, y\}, z)$ wherein $x$ and $y$ must be closer to each other in tree-distance than either is to $z$~\citep{vikram2016interactive}.
In the context of topic modeling, constraint-based interaction has typically focused on probabilistic models where models that violate a constraint are either down-weighted or eliminated. Whether these constraints are introduced all at once~\citep{andrzejewski2009incorporating, petterson2010word} or in interactive rounds~\citep{hu2014interactive}, the focus of these methods has been on modifying the prior distribution over topics so that they favor certain word correlations.

Higher-order feedback seeks to effect large changes in a model by modifying aspects of the model directly. As such, the types of feedback considered are highly dependent on the task at hand. In clustering, for example, researchers have considered split and merge requests in which a user indicates that a certain cluster ought to be broken up into smaller clusters (a split request) or that several clusters should be grouped together into a single cluster (a merge request)%. Given certain assumptions on the target clustering, upper bounds can be given on the number of rounds of interaction needed to find the target clustering
~\citep{balcan2008clustering, awasthi2014local}.

Perhaps the most convincing use of higher-order feedback in topic modeling is via anchor words. Because each anchor word has a unique topic associated with it, actions performed on anchor words have the potential to effect large changes in the topic model. ~\citet{lund2017tandem} proposed a protocol in which a user creates a group of words that they feel are representative of a topic and these words are aggregated into a single pseudo-anchor word. These pseudo-anchor words are then used to create a topic model in the same way that actual anchor words would be used.

The interactive protocol considered in this work is similar to that considered by~\citet{lund2017tandem} in its reliance on anchor words. However, our method differs considerably both in the types of words a user can interact with (we only allow a user to interact with geometrically-meaningful anchor words) and in the way we utilize the user-created groups (we sidestep the creation of pseudo-anchor words).

\begin{figure}[!ht]
\begin{enumerate}
	\item {\bf Compute normalized word co-occurrences.}
	Form the $V \times V$ matrix $\bar Q$, where $\bar Q_{ij} = \pr(w_2=j|w_1=i)$. The rows of $\bar Q$ lie in $\Delta_V$. 
	
	\item {\bf Identify the anchor words.}
	Find $K$ rows of $\bar Q$, say $s_1, \ldots, s_K$, such that the rest of the rows lie approximately in the convex hull of the $\bar{Q}_{s_i}$. These are the anchor words.
	
	\item {\bf Express all rows as convex combinations of anchor rows.}
	For each word $i$, find positive weights $C_{i,1}, \ldots, C_{i,K}$ summing to 1 such that $\bar{Q}_i \approx C_{i,1} \bar{Q}_{s_1} + \cdots + C_{i,K}\bar{Q}_{s_K}$. Then $C_{i,k} \approx \pr(z=k|w=i)$.
    \item {\bf Recover the topic distribution.}
    By Bayes' rule: $A_{i,k} = \pr(w=i|z=k) \propto C_{i,k} \pr(w=i)$.
\end{enumerate}
\caption{The generic anchor words algorithm.}
\label{alg:arora}
\end{figure}

\section{Preliminaries}\label{sec:preliminaries}

A \emph{corpus} is a collection of documents $d_1, \ldots, d_m$, each of which is represented in the bag-of-words representation as a vector in $\Z_{+}^V$, where $V$ is the size of the vocabulary. A \emph{word-topic matrix} is a $V \times K$ matrix $A$ such that each column $A_i$ corresponds to a topic and is represented as an element of $\Delta_V$, the $V$-dimensional probability simplex. 

Given a word-topic matrix $A$ and a prior distribution $\tau \in \Delta_K$, the generative model for a corpus is given as
\begin{itemize}
	\item For each document $d = 1, 2, \ldots$:
    \begin{itemize}
    	\item Draw a topic distribution $p_d \sim \tau$
        \item For word $i$ in document $d$, draw its topic $z_i \sim p_d$ and then draw the vocabulary word $w_i \sim A_{z_i}$.
    \end{itemize}
\end{itemize}

Together, the matrix $A$ and distribution $\tau$ induce a word co-occurrence matrix $Q \in \R^{V \times V}$ and topic co-occurrence matrix $R \in \R^{K \times K}$ satisfying
\begin{align*}
Q_{i,j} &= \pr(w_1 = i, w_2 = j) \ \ \text{ and}\\
R_{k, k'} &= \pr(z_1 = k, z_2 = k')
\end{align*}
for a randomly generated document with words $w_1$ and $w_2$ with associated topics $z_1$ and $z_2$.

We say that a word $i$ is an \emph{anchor word} for topic $k$ if $A_{i,k} \gg 0$ and $A_{i,k'} = 0$ for all other topics $k' \neq k$. Further, we say that the topic matrix is separable if each topic $k$ has an associated anchor word $s_k$.
 
Given such a corpus, several algorithms have been designed to provably recover the anchor words of a topic model and the topics associated with them \citep{arora2012learning, recht2012factoring, arora2013practical}. The general approach is given in Figure~\ref{alg:arora}. In this work, we will assume that we have access to such procedures and their subroutines.

\section{An anchor word based interactive protocol}\label{sec:interactive-protocol}

As pointed out above, there are many difficulties associated with topic modeling as a purely unsupervised task. These include the identification of the correct number of topics, filtering out noise, and dealing with the inherent ambiguities of language. Moreover, different users may have different desiderata in a topic model that may not be possible to satisfy simultaneously. 

To address these issues, several methods have been considered for injecting human knowledge into topic modeling. The approach with the closest resemblance to our own is the recently proposed \emph{anchor facet} approach \citep{lund2017tandem}. In this method, a user synthesizes pseudo-anchor words by averaging together subsets of words the user chooses. As we will see, these pseudo-anchors disregard the underlying geometry of the data in ways that can lead to problems in topic recovery.

The remainder of this section is organized as follows. We first give an example where the anchor facet approach leads to identifiability issues. Next, we present a generative model for which standard unsupervised techniques cannot recover the desired topics, even in the infinite data limit. Finally, we present our interactive protocol which can, in fact, find good estimates of the desired topics.

\subsection{An anchor facet problem}

In the anchor facet model, a user chooses a set of words $\G$ from the vocabulary that they feel should represent a topic. For instance, they might choose {\tt games} and {\tt computer} to indicate a `computer games' topic. The corresponding word co-occurrence vectors (rows of $\bar{Q}$) are then aggregated to form a \emph{pseudo-anchor} $g$, by taking their harmonic mean (among other options), and this $g$ is added to the set of anchor words. After the user has created the pseudo-anchors, a topic model is recovered using steps 3-4 of Figure~\ref{alg:arora}.

This approach is intuitively appealing but hard to justify geometrically. The correctness of the anchor words algorithm depends on the anchors being at the corners of the simplex containing all the word vectors. Pseudo-anchors violate this in two ways: (1) they don't have a clear meaning in terms of co-occurrence probabilities (if, as suggested, the harmonic mean is used for aggregation) and (2) they may well lie near the center of the simplex. For instance, it could easily happen that a large fraction of the remaining words are not well-approximated as convex combinations of pseudo-anchors; in which case, these words will be assigned to topics in a fairly arbitrary manner.

%\sanjoy{Could you guys change the hello world picture to a "computer", "videogames" example -- or something similar, that is a bit more realistic? We can talk about this. Also, I have commented out two paragraphs because I think they might be subsumed by the discussion above, but feel free to reinstate.}

\iffalse
Despite the intuitive appeal of synthesizing new anchors from groups of non-anchor words, difficulties can arise when grouping is done naively, as illustrated by Figure~\ref{fig:anchor facets}. In this example, a user, hoping to create a introductory programming topic, creates an anchor facet by merging the words `hello' and `world.' Because these words are relatively common, their associated  conditional distributions will be relatively uniform and thus close to the interior of the simplex spanned by the true anchor words. Thus any reasonable choice of aggregation will result in a point, called `hello-world,' somewhere in the middle of the simplex. 

Now any other word, say `cat' for example, may be written as the convex combination of the original anchor words $s_1$, $s_2$, and $s_3$ as well as the new point `hello-world.' However, because `hello-world' lies in the interior of the simplex, this way of writing the convex combination is no longer unique, and the particular choice can matter a great deal. Moreover, if some of the original anchors are removed, say anchor $s_3$ in our example, it may no longer be possible to guarantee that `cat' lies within or perhaps even near the new simplex.
\fi

\begin{figure}
    \begin{center}
	\includegraphics[scale=0.3]{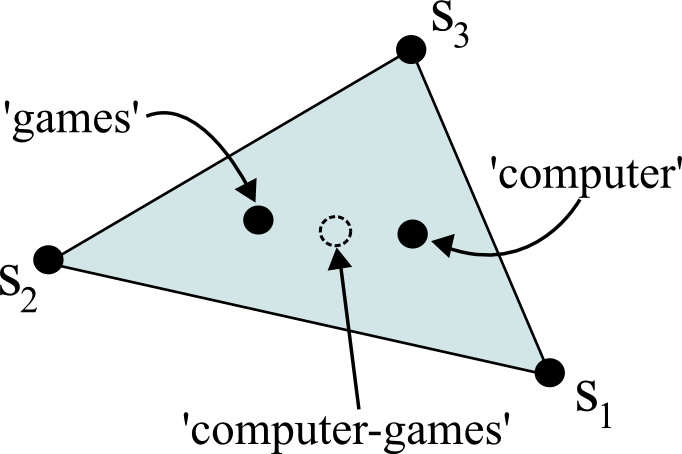}
	\end{center}
    \caption{Illustration of anchor facet shortcoming. Here the user combines anchor words `computer' and 'games' which results in a point `computer-games' somewhere in the middle of the simplex spanned by $s_1$, $s_2$ and $s_3$. \label{fig:anchor facets}}
\end{figure}

\subsection{A subtopic view of document generation}\label{sec:subtopic}

The topic modeling view of data generation, that a corpus is generated by a relatively small number of topics that are easily interpretable by a human, is often an oversimplification. In reality, documents on similar subjects can vary wildly in their choice of language due to authorship, the times they were written, etc. A topic model that accurately fits a real corpus must necessarily contain many topics.

To see this, imagine a corpus of news documents collected over the course of a year, in which a small but significant percentage of articles deal with weather. A user wishing to analyze this corpus via topic modeling might be satisfied with a single weather topic. However, the corpus itself will not look like it only has a single weather topic. Indeed, the distribution of words in a weather article written in September during hurricane season will look significantly different from the distribution of words in a weather article written in January during blizzard season, which in turn will differ from the distribution of words in a weather article written in July during drought season. 

Thus, accurately modeling the weather-related aspects of the corpus requires several topics. And that is just the weather! Conceivably, other aspects of the corpus for which a user might imagine a single topic sufficing can in turn be broken into components that actually model the data. On the other hand, a model with hundreds or thousands overlapping and highly correlated topics is not easy to work with. Many users would prefer a significantly simpler model that may not perfectly describe the data but summarizes the core subjects well.

To model the mismatch between the idealized view a user has in mind and the actual data generating process, we introduce the \emph{subtopics view} of data generation. It is described by the following generative model.

\begin{itemize}
	\item There are several `ideal' topics $M_1, \ldots, M_K$ along with some topic-topic co-occurrence matrix $R^{M} \in \R^{K \times K}$.
    \item For each ideal topic $M_k$, some number of `subtopics' indexed by the set $\G_k$ are drawn i.i.d. from a distribution satisfying $\E[A_t] = M_k$ for each $t \in \G_k$.
    \item The corpus is generated according to the new topic matrix $A$ and some topic-topic co-occurrence matrix $R^{A}$ satisfying $\sum_{t \in \G_k} \sum_{t' \in \G_k'} R^A_{tt'} = R^M_{kk'}$.
\end{itemize}

Here we call the topic model induced by $M$ and $R^M$ as the \emph{idealized model} and the topic model induced by $A$ and $R^A$ the \emph{subtopics model}. Intuitively, the idealized model is the model that would have generated the corpus in an ideal world. However, the corpus is actually generated by the subtopics model with a larger number of more specific topics. The following lemma, whose proof is omitted for space, demonstrates that the co-occurrence matrix induced by a subtopic model is intrinsically biased away from the idealized model in expectation.

\begin{lemma}
If $Q^M$ is the co-occurrence matrix induced by the idealized model and $Q^A$ is the co-occurrence matrix induced by the subtopics model, then \[ \E_A[Q^A] = Q^M + \sum_{k}R^M_{k,k} \Sigma^{(k)} \]
where $\Sigma^{(k)}$ is the covariance matrix of the subtopic distributions generated under ideal topic $k$.
\end{lemma}

Thus, in general the co-occurrence matrix generated by the subtopics model is biased away from the co-occurrence matrix that would be generated by the idealized model. %Indeed, in the special case where $\Sigma^{(k)} = \Sigma$ for $k=1, \ldots, K$, the above reduces to
%\[ \E_A[Q^A] = Q^M + \tr(R^M) \cdot \Sigma . \]
Thus, directly fitting a topic model based on these statistics should not in general recover the ideal topics. Rather, some other approach is needed.

%and finally average groups of these to get estimates of the ideal topics $M$.

%Instead of trying to directly recover the idealized model $M$, we could look to first recover the subtopics $A$ and then, by leveraging a user's knowledge, average groups of these to get estimates of the ideal topics $M$. 

\subsection{An interactive protocol}\label{sec:protocol}

How do we recover the idealized topics $M$? Returning to our weather example, we could start by fitting a model with say, 500 topics. Next, we could ask a user to peruse these, form a group of some good weather subtopics e.g. hurricane, blizzard, drought, etc., and then average subtopics in the group to get an estimate of an ideal weather topic. But the way the topics are displayed presents a challenge: perusing 500 topics and finding their salient groupings might place an overwhelming cognitive load on a user. Indeed, even if each subtopic is uniquely identified by its top 10 words (which often is not the case), a prospective user would have to wade through 5000 words! What is needed, then, is a way to ensure we have a unique representation for each topic and to present these to the user as succinctly as possible.

Our approach is to utilize anchor words. Assuming each subtopic is associated with an anchor word, we find an `overcomplete' list of anchor words $s_1,\dots,s_T$ and present these to a user as proxies for entire topics. The user can quickly sort through this list and easily identify subtopics by their component anchor words. After a few rounds, the user will form $K$ groupings of selected anchor word indices $\widehat{\G}_1, \ldots, \widehat{\G}_K \subset \{1, \ldots, T\}$. %It is possible that there are anchor words that a user simply does not recognize as significant, perhaps because they are uninteresting background topics. Thus, we do not require a complete partition of the anchor words from the user.
      
%Given a corpus generated by a subtopic model, our interactive protocol for recovering estimates of the idealized model is relatively simple.
%\begin{itemize}
%	\item Run an anchor word-based topic modeling algorithm to recover a large number of anchor words $s_1, \ldots, s_T$ and their associated topic vectors $\widehat{A}_1, \ldots, \widehat{A}_T$.
 %   \item Present the anchor words to the user and receive $K$ groupings of selected anchor word indices $\widehat{\G}_1, \ldots, \widehat{\G}_K \subset \{1, \ldots, T\}$.
%    \item For each group $\G_k$, average the associated topic vectors $\widehat{M}_k = \frac{1}{|\widehat{\G}_k|} \sum_{j \in \widehat{\G}_k} \widehat{A}_i$.
%\end{itemize}
%It is conceivable that there are anchor words which a user simply does not recognize as significant, either because they represent background or garbage topics or because they were mistakenly identified as anchor words due to noise. Thus, we do not require that the user give us a complete partition of the anchor words. That is, we do not require $\bigcup_k \widehat{G}_k = \{1, \ldots, T\}$.

%It is not hard to see that if (a) our estimates of the subtopics $A$ are unbiased and accurate and (b) the user makes an unbiased selection of some subset of each true subtopic group $\G_k$, then each estimate $\widehat{M}_k$ will be close to the ideal topic $M_k$.

%\stefanos{

Given a corpus generated by the subtopic model, our interactive protocol for estimating the $M_k$'s is relatively simple and it is given in Figure~\ref{alg:averaging}. It is not hard to see that if our estimates of the subtopics $A$ are unbiased and the user correctly identifies each true subtopic group $\G_k$, then each estimate $\widehat{M}_k$ will be close to the ideal topic $M_k$.
\iffalse
\begin{itemize}
	\item Run an anchor word algorithm to find `candidate' anchor words $s_1,\dots,s_T$. 	
    \item Present the anchor words to the user and receive $K$ groupings of selected anchor word indices $\widehat{\G}_1, \ldots, \widehat{\G}_K \subset \{1, \ldots, T\}$. 
    \item Let $C= \widehat \G_1 \cup \cdots \cup \widehat{\G}_K$ denote the set of indices of anchor words that the user selected. 
    \item For every $j\in C$, estimate their associated topic vectors $\widehat A_{1},\dots,\widehat A_{|C|}$ by running a topic recovery algorithm. 
    \item For each group $\widehat\G_k$, average the associated topic vectors $\widehat{M}_k = \frac{1}{|\widehat{\G}_k|} \sum_{j \in \widehat{\G}_k} \widehat{A}_j$.
\end{itemize}
\fi

\iffalse
\begin{itemize}
	\item[(a)] Identify the `candidate' anchor words $s_1,\dots,s_T$ via a standard anchor-finding algorithm.	
    \item[(b)] Present these to the user and receive $K$ groupings of selected anchor word indices $\widehat{\G}_1, \ldots, \widehat{\G}_K \subset \{1, \ldots, T\}$. 
    \item[(c)] Using all the anchor words, estimate the topic vectors $\widehat A_{1}, \ldots, \widehat{A}_T$ by running a topic recovery algorithm.
    \item[(d)] For each group $\widehat\G_k$, average the associated topic vectors $\widehat{M}_k = \frac{1}{|\widehat{\G}_k|} \sum_{j \in \widehat{\G}_k} \widehat{A}_j$.
\end{itemize}
\fi

Issues arise when, due to undersampling, the set of candidate anchor words contains words that are not true anchor words. These `spurious' anchor words disrupt our ability to estimate the subtopics, leading to errors in our estimates of the ideal topics. To counter this issue, we consider an alternative procedure that replaces step (c) with the following:
\begin{itemize}
    \item[(c')] Using only the anchor words selected by the user, estimate the topic vectors $\widehat A_{j}$ for each $j \in \G_{1} \cup \cdots \cup \G_{K}$ by running a topic recovery algorithm.
\end{itemize}

We call the algorithm that uses step (c') \emph{partial interactive recovery} to distinguish it from the \emph{full interactive recovery} algorithm that uses step (c).

\begin{figure}[ht]
\framebox[3.25in]{
\begin{minipage}[t]{3in}
\vskip.05in
\begin{itemize}
	\item[(a)] Identify the `candidate' anchor words $s_1,\dots,s_T$ via a standard anchor-finding algorithm.	
    \item[(b)] Present these to the user and receive $K$ groupings of selected anchor word indices $\widehat{\G}_1, \ldots, \widehat{\G}_K \subset \{1, \ldots, T\}$. 
    \item[(c)] Using all the anchor words, estimate the topic vectors $\widehat A_{1}, \ldots, \widehat{A}_T$ by running a topic recovery algorithm.
    \item[(d)] For each group $\widehat\G_k$, average the associated topic vectors $\widehat{M}_k = \frac{1}{|\widehat{\G}_k|} \sum_{j \in \widehat{\G}_k} \widehat{A}_j$.
\end{itemize}
\vskip.05in
\end{minipage}}
\caption{Full interactive recovery algorithm}
\label{alg:averaging}
\end{figure}

\section{Document representation experiment}\label{sec:experiments}
%In this section, we study real and simulated users in two experiments. First, we look at a real dataset and simulate a user seeking to produce a topic model that results in meaningful document representations, that is, documents that share similar subjects should have similar representations. Next, we explore if real users, equipped with our interactive tools, can understand the main aspects of the corpus that they are analyzing and can create topic models that are interpretable. 

To compare the quality of the topic models produced by the various algorithms, we conducted an experiment on the inferred document representations produced by these models. We used the \textbf{20 Newsgroups} dataset,\footnote{http://qwone.com/~jason/20Newsgroups} which consists of $\approx 18$K documents each belonging to one of 20 categories. 

We again compared our anchor group approach against the anchor facet approach of~\citet{lund2017tandem} and the constraint-based approach of~\citet{hu2014interactive}. We ran the anchor finding algorithm of~\citet{arora2013practical} to generate 500 candidate anchor words. For the interactive anchor-based approaches, we calculated
\[ g(a, c) = \frac{\# \text{ times } a \text{ occurs in document with label } c}{\# \text{ times } a \text{ occurs in corpus}}\]
for each anchor word $a$ and each news group category $c$; and for each category $c$, we selected the 10 anchors words $a$ with the highest $g(a,c)$ value.

For the constraint-based approach, we calculated $g(a, c)$ for all words $a$, not just anchor words, and selected the 10 words $a$ with the highest $g(a,c)$ value for each category $c$. For the resulting grouping, we generated all of the corresponding SPLIT and MERGE constraints.

%We note that our goal in this experiment is not to achieve state of the art classification performance but rather to conduct an analysis that reflects the needs of a real user of a topic model: documents about similar subjects ought to be closer together than documents about different subjects. 
To evaluate the quality of the competing topic models we looked at the local neighborhood structure of the resulting document representations using a k-nearest neighbor (k-NN) classifier. For a given topic model with $m$ topics, we embedded the documents into the $m$-dimensional probability simplex using LDA~\cite{GS04}. We then computed the leave-one-out cross-validation (LOOCV) accuracy of the k-NN classifier over a sample of 2K embedded documents.
%We compared the performances of both interactive recovery methods as well as topic models using all 500 anchors, the 200 selected anchors, and a baseline model with 20 topics. For all methods we used \textsc{RecoverL2} of~\citet{arora2013practical}. 
Table~\ref{tab:nearest neighbor results} presents the performances of the resulting k-NN's for varying values of k on several interactive and non-interactive methods. All interactive methods had 20 topics (one for each news group category), whereas the number of topics varied for the non-interactive ones. 

\begin{table}
\centering
\begin{tabular}{lrrrrr}
  \hline
Model & k = 10 & k = 20 & k = 50 & k = 100  \\ 
  \hline
{\sc full 20} & 0.330 & 0.324 & 0.309 & 0.273 \\ 
  {\sc partial 20} & \textbf{0.337} & \textbf{0.337} & \textbf{0.321} & \textbf{0.287} \\ 
  {\sc Lund et al. 20} & 0.236 & 0.223 & 0.197 & 0.173 \\ 
  {\sc Hu et al. 20} & 0.221 & 0.212 & 0.196 & 0.178 \\ 
   \hline
  {\sc all 500} & 0.218 & 0.193 & 0.155 & 0.126 \\ 
  {\sc select 200} & 0.228 & 0.199 & 0.158 & 0.130 \\ 
  {\sc vanilla 20} & 0.144 & 0.140 & 0.133 & 0.121 \\   \hline
\end{tabular}
\caption{Accuracy of k-NN classifiers under the various topic recovery algorithms. {\sc all 500} refers to the topic model using all 500 anchor words. {\sc select 200} refers to the topic model using the 200 chosen anchor words. {\sc vanilla 20} refers to a model with 20 topics. \label{tab:nearest neighbor results}}
\end{table}

We observed that for all values of k, our interactive algorithms ({\sc Full} and {\sc partial}) outperformed all other interactive and non-interactive approaches.
%, even though two of those non-interactive models ({\sc all 500} and {\sc select 200}) result in much higher-dimensional embeddings. Also, partial recovery outperformed full interactive recovery.
%but not by much. 

\section{User study}\label{sec:user_experiments}
We conducted a small-scale user study to evaluate the anchor group interactive algorithm. Five users were asked to create their own topic model based on a corpus of recent news articles. All users were doctoral students in computer science, three of whom had past experience with topic modeling. 

\iffalse
We begin by describing the dataset that was used in the study as well as the preprocessing we performed. We then describe the interactive process under which we collected user feedback from these candidate anchor words. Finally, we compare the performance of our estimator to that of the algorithm given in \cite{arora2013practical}.      
\fi

\subsection{Data collection and preprocessing}
We used a collection of news articles crawled from the \textbf{CNN} website as its corpus; it was provided to us by a commercial search engine. The corpus contained about 10K articles, starting from around April 2016 and spanning about year. The articles covered a diverse range of subjects including politics, economy, sports, technology, science, and law. It also spanned several notable events such as the 2016 U.S. presidential debates and election, the 2016 Olympics games, and the Brexit referendum. It is also worth noting that since the dataset was created by a crawler, some articles contain boilerplate content such as advertisements and links to other irrelevant articles, which we did not take any steps to remove. We also did not perform any stemming. We only removed stop words and kept words that occurred in at least 10 documents. The final vocabulary contained about 17K words. After running an anchor word algorithm~\cite{recht2012factoring}, we had a list of about 500 anchor words as the basis of our interactive interface.

%Roughly, given a non-negative data matrix $X$, the algorithm first finds a matrix $W$ that contains the \emph{hott rows} of $X$. The algorithm then proceeds to finding an approximate non-negative matrix factorization $X=AW$. In our experiment, we generated candidate anchor words by finding the `hott rows' of the empirical normalized word co-occurrence matrix $\bar Q$. It is not hard to see that the `hott rows' of $\bar Q$ are also the `hott rows' of the original word-document matrix $X$. 
%After estimating candidate anchor words, 
%we used the \textsc{RecoverL2} procedure of \cite{arora2013practical} to estimate subtopic vectors $\widehat A_{1},\dots, \widehat A_{500}$.    

\subsection{Interactive process}
User feedback was collected via a web-based interface. At the beginning, users were prompted to select an element from the list of anchor words. After a word was selected, the user was taken to a separate screen where they created a topic by grouping words they felt were similar enough to the originally selected word. This component of the interface made topic creation more efficient by reducing the number of anchor words a user scanned to create a group. This process was repeated until the user felt they had exhausted the list of all salient groups.

Before starting the process, users were given some brief information about the dataset and then asked to create topics that would best summarize it. 

\iffalse
Perusing a list of 500 words many times can be taxing on a user. To help users better traverse the space of anchor words, the interface had four additional features. 
\begin{itemize}
\item \emph{Complete topic}: After merging anchors into a topic, the user could complete the topic with anchors suggested by the system. The suggested words were sorted by $\ell_1$ distance.
\item \emph{Merge topics}: The user was given the option to merge two or more created topics into one. 
\item \emph{Delete topic}: The user was given the option to delete a grouping they had created.
\item \emph{Suggest topics}: When creating a new topic the user was given the option to hit a button that suggested new anchors. The system highlights words that are further away in $\ell_2$ distance from the space spanned by the words already selected by the user. 
\end{itemize}
\fi
%that is trying to create a topic about the hacking interference during the presidential elections. The user has initially selected the anchor `hackers' and at this point, she/he has decided to create a topic that also will contain the anchors `fbi', `emails', and `messages'. The box to the right displays a `suggestion' of anchors. $l_1$ distance to the set of words that are currently in the topic.

%\begin{figure*}
%\subfigure[Metrics]{\includegraphics[width=0.5\textwidth]{figures/metrics}}
%\subfigure[Word intrusion]{\includegraphics[width=0.5\textwidth]{figures/intrusion}}
%\caption{Different metrics from the topic models used in t\label{fig:result}}
%\end{figure*}

\begin{figure*}
    \begin{center}
	\includegraphics[width=0.7\textwidth]{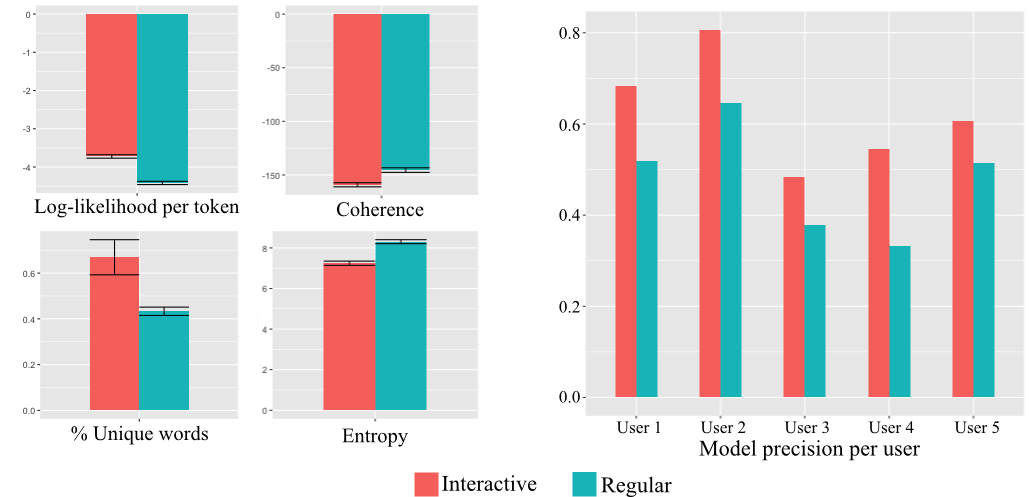}
	\end{center}
	\vspace{-1em}
    \caption{Comparison of models built using user interaction with their non-interactive counterparts. \emph{Left}: Log-likelihood per token, coherence, \% unique words, and average entropy of topics. \emph{Right}: Per-user performance on word intrusion task. Users were tested on all user-created topics, not just their own. \label{fig:result}}
\end{figure*}

\subsection{Results}
Using the interactive process described above, users created groupings of anchor words. Some examples of these groupings are given in Table~\ref{tab:groupings}.

\begin{table}[ht]
\centering
\small
\begin{tabular}{rl}
\hline
1& russian putin intelligence agencies \\ \hline
2& olympics rio olympic athletes brazil sport winner \\ \hline
 3& hollywood movie entertainment star film \\    
    & character original awards controversy \\ \hline
 4& joe politics vice rubio cruz kasich ballot  \\
    & campaigns convention voting poll delegates \\ 
    & elections pennsylvania \\ \hline
 5& israel peace region council terrorist terror isis \\
   & suicide iraqi falluja iraq troops syrian syria \\
   & aleppo refugees turkey \\ 
   \hline
\end{tabular}
\caption{Examples of user anchor groupings.\label{tab:groupings}}
\end{table}

After collecting the feedback that users provided, we used the partial interactive recovery algorithm of Section~\ref{sec:protocol} with the \textsc{RecoverL2} of~\cite{arora2013practical} to learn a topic model for each user. We call models created by user interactions \textbf{Interactive}. For each user, we also learned a topic model with the same number of topics without any interaction using Algorithm 1 from \cite{arora2013practical}. We call these models \textbf{Regular}. 

\subsubsection{Qualitative assessment of topics}
Shaded rows of Table~\ref{tab:topics} give the most probable words under topics learned using the user feedback of Table~\ref{tab:groupings}. Unshaded rows show the most probable words under the topic of the regular method that was closest in $\ell_{1}$ distance to the one above it. 

\bgroup
\setlength\tabcolsep{0.2em}
\begin{table}[ht]
\centering
\small
\begin{tabular}{lll}
  \hline
 \rowcolor{light-gray}1 & \textbf{Interactive} & russian putin russia intelligence obama \\ 
 &\textbf{Regular}  & obama president trump clinton visits \\ \hline
 \rowcolor{light-gray}2& \textbf{Interactive} & rio olympic olympics games athletes \\ 
 & \textbf{Regular}  & minister prime company million published \\ \hline
 \rowcolor{light-gray} 3& \textbf{Interactive} & film star show awards disney \\ 
 &\textbf{Regular}  & trump comedy show company million \\ \hline
\rowcolor{light-gray} 4& \textbf{Interactive} & cruz kasich president clinton convention \\ 
 & \textbf{Regular}  & trump clinton donald campaign trumps \\ \hline
\rowcolor{light-gray} 5 & \textbf{Interactive} & falluja isis battle syrian forces \\ 
 & \textbf{Regular}  & attacks brussels terror airport police \\
   \hline
\end{tabular}
\caption{Most probable words for the user created topics shown in Table~\ref{tab:groupings}.\label{tab:topics}}
\end{table}
\egroup

Across the board, the interactive method resulted in better quality topics that seemed to align with the intentions of the user that created them. Moreover, interactive topics seemed more easily interpretable and more general than the topics of the regular method. %For example, looking at topics 1 and 2 in Table~\ref{tab:topics}, one can see that the interactive method yielded topics that matched what the user was trying to achieve. (See groupings 1 and 2 in Table~\ref{tab:groupings}.) We observe a similar situation for topic 5, for which the interactive method yielded a topic related to current events in the Middle East, while the regular method yielded a very specific topic about the Brussels terror attack.

\subsubsection{Word intrusion user evaluations}
As noted in the introduction, a popular way to understand the gist of a topic is to look at its $n$ most probable words and try to find their common theme. \emph{Word intrusion} seeks to quantify how easily one can interpret a topic model in this way~\cite{chang2009reading}. Roughly, for each topic, its list of $n$ most probable words will be \emph{intruded} by a word that is in the $n$ most probable words of another topic. Humans are asked to find the \emph{intruding} words and models are then scored according the \% of intruding words found by humans. One would expect that in a semantically coherent list of words, intruding words will be more easily detected.

To measure word intrusion, each user that participated in the study was asked to evaluate a mix of their own and of other users' topics, as well as the topics of the regular method. The number of words that were shown was $n=10$. Figure~\ref{fig:result} (right) shows the results of this experiment. Across the board, users performed better on the word intrusion when they were evaluating an interactive topic as opposed to one found by purely unsupervised methods, even when those interactive topics were created by other users.

\subsubsection{Quantitative metrics}
We also compared the two methods across different metrics. We looked at log-likelihood, semantic coherence, which was introduced by~\citet{mimno2011optimizing}, proportion of unique most probable words, and entropy. To calculate log-likelihood, we ran 100 iterations of the Gibbs sampler while keeping the topics of each method fixed. Figure~\ref{fig:result} (left) shows the different metrics. Averaged across users, the interactive method has slightly higher per token log-likelihood but slightly worse topic coherence at the top $n=10$ words. Also, the interactive method has more unique most probable words per topic (again for $n=10$), indicating models that capture topics that are different from each other. Finally, the interactive method has lower entropy, indicating that on average, its topics concentrate on a smaller number of words than the regular method.   

%\bibliography{bibliography}
\bibliographystyle{icml2019}

\end{document}